\documentclass[twoside,leqno,twocolumn]{article}

\usepackage[letterpaper]{geometry}

\usepackage{siamproceedings}

\usepackage[T1]{fontenc}
\usepackage{amsfonts}
\usepackage{graphicx}
\usepackage{epstopdf}
\usepackage{enumitem}
\usepackage{algorithmic}

\usepackage{booktabs}
\usepackage{multirow}
\usepackage{tabularx}
\usepackage{threeparttable}
\usepackage{subcaption}
\usepackage{amsmath}
\usepackage{array}
\usepackage{xcolor}
\usepackage{amsmath}
\usepackage{xltabular}
\usepackage{booktabs}
\usepackage[dvipsnames]{xcolor}
\usepackage{hyperref}

\hypersetup{
  pdftitle={Spectral Analysis of Fake News Propagation},
  pdfauthor={Weibin Cai and Reza Zafarani},
  pdfkeywords={fake news propagation, spectral graph theory, spectral analysis}
}

\ifpdf
  \DeclareGraphicsExtensions{.eps,.pdf,.png,.jpg}
\else
  \DeclareGraphicsExtensions{.eps}
\fi

\newsiamremark{remark}{Remark}
\newsiamremark{hypothesis}{Hypothesis}
\crefname{hypothesis}{Hypothesis}{Hypotheses}
\newsiamthm{claim}{Claim}

\usepackage{amsopn}

\begin{document}

\newcommand\relatedversion{}

\title{\Large Spectral Analysis of Fake News Propagation\relatedversion}

\author{
Weibin Cai\thanks{Data Lab, Department of EECS, Syracuse University, New York, USA. Emails: \email{weibin44@data.syr.edu}, \email{reza@data.syr.edu}.}
\and Reza Zafarani\footnotemark[1]
}


\date{}

\maketitle


\fancyfoot[R]{\scriptsize{Copyright \textcopyright\ 2026 by SIAM\\
Unauthorized reproduction of this article is prohibited}}





\begin{abstract} 
How can we systematically represent the propagation of information? The propagation structure of fake news has been shown to be an important cue for detecting it; yet, existing propagation-based fake news detection methods have mainly relied on ad hoc topological features, and a unified view of cascade patterns is still lacking. To address this, we study news propagation from a spectral view by connecting graph spectra to propagation-related structural properties through rigorous spectral bounds. We introduce several new bounds and integrate them with existing bounds into a unified spectral representation of information propagation. We then use these spectral bounds for downstream classification and design a discrete structural optimization framework to interpret learned propagation patterns. For efficient optimization, we rely on a first-order perturbation approximation and consider both score-guided and bound-guided objectives. Experiments on real-world data reveal meaningful spectral differences between fake and real news, competitive classification performance, and interpretable evolution trajectories from structural optimization. The findings demonstrate the value of spectral analysis for understanding and modeling information propagation.

\end{abstract}

\section{Introduction.}

Fake news detection is an important problem because false content can spread rapidly on social media and cause harmful social impact~\cite{lazer2018science,bakdash2018future,vosoughi2018spread}. A natural way to detect fake news is to examine its content or source. Existing early fake news detection methods mainly fall into three categories: \textit{knowledge-based} approaches, which verify claims against external knowledge bases~\cite{shi2016discriminative,ciampaglia2015computational}; \textit{style-based} approaches, which exploit linguistic and stylistic cues in news text~\cite{zhou2019fake,zhou2020survey}; and \textit{source-based} approaches, which assess the credibility of the entities involved in news creation, publication, or dissemination~\cite{horne2018assessing,norregaard2019nela}. However, these methods all have inherent limitations. Knowledge-based methods depend on reliable external knowledge, which is often incomplete or unavailable in practice~\cite{shu2017fake}. Style-based methods are vulnerable to intentional style manipulation~\cite{zhou2020survey}, a problem that has become more severe with the rise of generative language models~\cite{kumar2025peeping}. Source-based methods provide only indirect signals, since source credibility does not necessarily imply the veracity of a specific news item. As a result, relying on source credibility alone may lead to overly arbitrary judgments~\cite{zhou2020survey}. These methods provide useful signals from the perspectives of \textit{what the news says} and \textit{who is involved in it}, but they do not capture the structural characteristics and diffusion signals of \textit{how the news propagates} on social networks. Prior \textit{propagation-based} studies have shown that fake and real news exhibit different propagation patterns on social networks. For example, fake news tends to spread faster and deeper, and attracts more strongly connected spreaders~\cite{vosoughi2018spread,zhou2019network}. These findings suggest that propagation structure provides complementary and robust signals for fake news detection, and has inspired a growing body of structural and graph-based methods~\cite{ma2018rumor,bian2020rumor,zhou2019network,shu2020hierarchical,wu2015false}.

\begin{figure}[t]
    \centering
    \includegraphics[width=\linewidth]{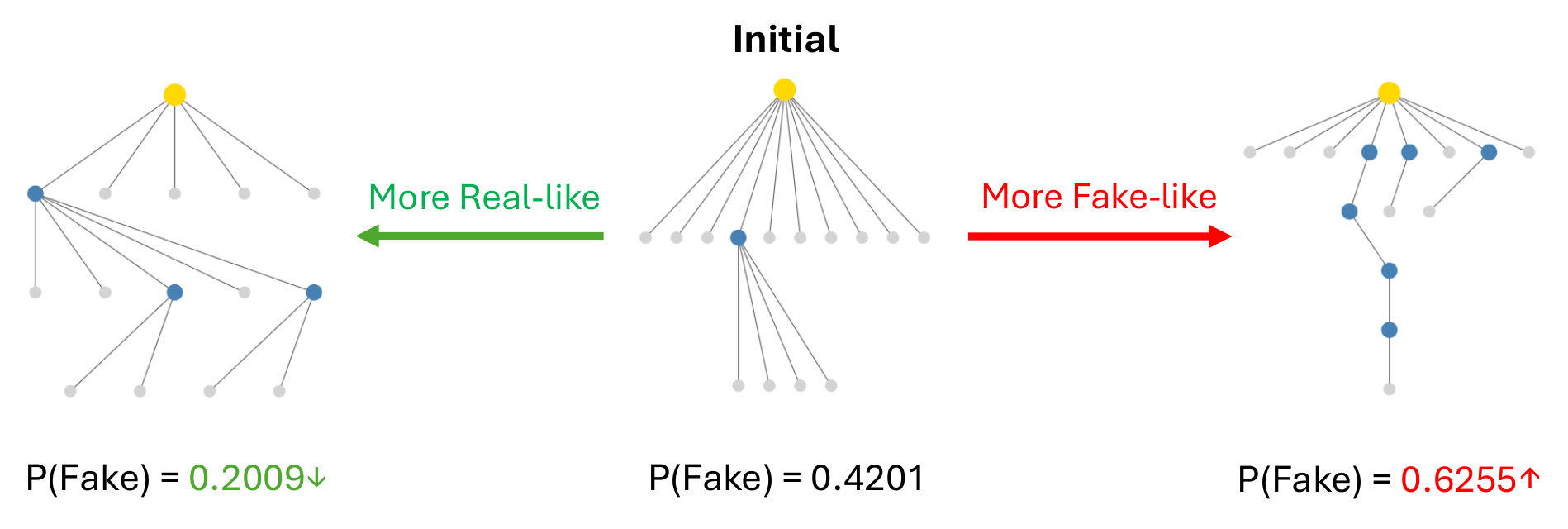}\vspace{-2mm}
    \caption{\small The benefits of spectral analysis of news. Illustration of the designed structural optimization from the same initial news propagation tree (nodes are spreaders and edges are spreading patterns). Optimizing (after some steps) toward a more ``\textcolor{red}{fake}-like" pattern (higher fake probability) produces a \textcolor{red}{deeper} and more \textcolor{red}{imbalanced} structure, while optimizing toward a more ``\textcolor{ForestGreen}{real}-like" pattern yields a \textcolor{ForestGreen}{shallower} and more \textcolor{ForestGreen}{balanced} propagation tree.}\vspace{0mm}    
    \label{fig:intro_example}
\end{figure}

Existing propagation-based methods can be broadly divided into two lines of work. The first line focuses on \textit{neural representation learning} for propagation structures, where graph neural networks~\cite{kipf2016semi} and related techniques~\cite{you2020graph} encode news propagation graphs into latent representations~\cite{ma2018rumor,bian2020rumor,wei2022uncertainty,cui2024propagation}. Although these methods often achieve strong classification performance, the learned representations are difficult to interpret, since individual dimensions usually lack clear structural meaning and the latent space remains hard to analyze directly. The second line uses \textit{handcrafted topological features} to represent propagation graphs~\cite{vosoughi2018spread,zhou2019network,shu2020hierarchical,kwon2013prominent}. Prior studies have shown that fake news exhibits distinct propagation behaviors from real news through features such as depth, cascade size, and structural virality~\cite{vosoughi2018spread}. While these handcrafted features provide intuitive insights, they are largely ad hoc and capture only specific aspects of cascade structure rather than providing a unified view of all propagation patterns.

To better understand propagation mechanisms, it is desirable to develop a unified representation of information cascades that can captures global structural patterns, their variations, and help distinguish between different types of cascades (e.g., fake or truth). Spectral graph theory provides a natural path for analyzing propagation graphs. While two cascades may appear similar in size or some topological features, such as depth and structural virality, their underlying ``spectral fingerprints'', which are derived from the eigenvalues of matrices such as the adjacency or Laplacian often reveal fundamental differences in how the information was propagated (see Figure \ref{fig:spectral-example}). These spectra capture global structural properties such as degree distribution, clustering coefficient, connectivity, and the like~\cite{jin2020spectral}. Moreover, many cascade characteristics, e.g., breadth, diameter~\cite{guo2025laplacian}, can be linked to eigenvalues through well-established bounds, offering a unified and interpretable representation of propagation structures.\vspace{1mm}

\noindent \textbf{Present work: Spectral Analysis of Fake News.} Motivated by these observations, we study news propagation from a spectral perspective. The key idea is to connect graph spectra with propagation-related structural properties through spectral bounds, and use this connection to analyze how fake and real news differ in their propagation patterns. Under this perspective, we develop two components for propagation analysis. \underline{(1) Spectral Representation}: We represent propagation patterns through five general aspects that describe their structure: branching capacity, cascade scale, structural cohesion, propagation span, and diffusion dynamics; each aspect encompasses various cascade properties, which are in turn linked to numerous spectral bounds. Together, these bounds form a unified structural representation of propagation patterns. We further analyze their tightness and correlation with the corresponding graph properties, and test their utility by using them in downstream classification. \underline{(2) Structural Optimization}: We further develop a discrete structural optimization framework to study how propagation patterns evolve under different scenarios, e.g., fake or real. This optimization process, using a first-order perturbation approximation, allows us to trace how propagation structures evolve toward a target direction. The result provides insight into the propagation patterns associated with fake and real news. As illustrated in Figure~\ref{fig:intro_example}, the same initial propagation tree can evolve toward more fake-like or more real-like structures under our optimization framework, revealing clear structural differences between the two. In summary, our contributions are as follows:

\begin{itemize}
    \item \textbf{Spectral analysis of information propagation graphs.} We develop a spectral view of news propagation by connecting graph spectra with propagation-related structural properties through spectral bounds. In particular, we prove various new bounds and unify them with existing results to form a propagation-oriented spectral representation. Based on this view, we analyze the differences between fake and real news in their spectra, and study the tightness and correlation of spectral bounds with the corresponding graph properties. 
    \vspace{-1mm}
    \item \textbf{Spectral representation of propagation graphs.} We use spectral bounds as a unified structural representation of propagation graphs. To the best of our knowledge, this is the first work to represent news propagation graphs from a spectral perspective. We further validate that these bound-based representations provide informative structural signals for fake news classification.
    \vspace{-1mm}
    \item \textbf{Structural optimization of propagation patterns.} We propose a discrete structural optimization framework to interpret propagation patterns in the space of graphs. Through this optimization, we obtain interpretable structural evolution trajectories and gain deeper insight into the differences between fake and real news propagation.
    \vspace{-1mm}

\end{itemize}
\section{Related Work.} As we focus on the propagation structure of fake news, we concentrate on representations that model how news spreads on social networks.


\subsection{Neural Representation Learning.}
Current propagation-based methods can be broadly divided into two types: (1) \textit{homogeneous propagation} models news diffusion using graphs whose nodes are posts and whose edges represent replies or retweets. Representative models such as RvNN~\cite{ma2018rumor} and BiGCN~\cite{bian2020rumor} build such propagation graphs and use RNNs~\cite{cho2014properties} or GNNs~\cite{kipf2016semi} to learn post representations. Some studies further incorporate temporal information, such as posting time, to capture finer-grained diffusion dynamics~\cite{sun2022ddgcn,peng2024rumor}. Since social media responses may be noisy or intentionally misleading, recent work has also improved robustness through graph contrastive learning~\cite{he2021rumor,ma2022towards,sun2022rumor,cui2024propagation,zhang2024evolving} and adversarial training~\cite{song2021adversary,sun2022rumor}. In contrast, (2) \textit{heterogeneous propagation} incorporates richer social context by allowing multiple node and edge types, such as news, posts, and users, together with relations such as posting, replying, retweeting, or following~\cite{Monti2019FakeND,shu2020hierarchical}. HPFN~\cite{shu2020hierarchical} investigates propagation features from structural, temporal, and linguistic perspectives between fake and real news through macro-level (where network includes news nodes, tweet nodes, and retweet nodes) and micro-level (where networks are conversation trees represented by cascade of reply nodes) networks. The authors further incorporate GCNFN~\cite{Monti2019FakeND} into this framework to examine the contribution of different feature types.

\subsection{Handcrafted Topological Features.} This line of work uses manually designed topological features to characterize propagation cascades. For example, Vosoughi et al.~\cite{vosoughi2018spread} measure static features such as depth, cascade size, maximum breadth, and structural virality~\cite{goel2016structural}, together with their temporal evolution, and show that fake news tends to spread deeper, faster, and more broadly than real news. Shu et al.~\cite{shu2020hierarchical} extract the following four topological features in both  macro and micro networks: (1) depth; (2) cascade size; (3) maximum outdegree; and (4) depth of the node with the maximum outdegree in macro-level network. They use $t$-tests to show whether differences between fake news and real news are significant. Zhou et al.~\cite{zhou2019network} categorize propagation-related features into four patterns: (1) more-spreader; (2) farther-distance; (3) stronger-engagement; and (4) dense-networks. Notably, these features are derived from the friendship network of news spreaders rather than the news propagation structure itself. Their analysis shows that fake news tends to spread farther and involve users who are more densely connected and more highly engaged.

\section{Preliminaries and Notation.}
For an undirected graph $G=(V, E)$ with vertices $V=\{v_1, v_2, \cdots,v_n\}$ and edges $E \subseteq V \times V$, its adjacency matrix $A \in \mathbb{R}^{n \times n}$ has $A_{ij} = 1$ if $(i, j) \in E$, otherwise 0. The degree matrix $D \in \mathbb{R}^{n \times n}$ is a diagonal matrix with node degrees on its diagonal, i.e., $D_{ii} = \sum_{j=0}^{n} A_{ij} $. Under the news propagation scenario, each node represents a comment or retweet. Consequently, the propagation graph contains no cycles and is therefore a tree. Although the graph structure is simplified from a general graph to a tree, many graph properties (e.g., diameter) remain meaningful, particularly when the graph is represented using its spectrum. The spectrum of a matrix is the set of its eigenvalues, which implicitly reveals its propagation structure properties. A graph can be represented by different matrices constructed from the adjacency matrix $A$ and degree matrix $D$, such as the transition matrix, and normalized Laplacian; their spectra capture distinct structural characteristics of the graph. Here, we focus on three matrices: (1) adjacency matrix $A$ with spectrum $\lambda_1 \ge \lambda_2 \ge \cdots \ge \lambda_n$; (2) laplacian matrix $L = D - A$ with spectrum $\mu_1 \ge \mu_2 \ge \cdots \ge \mu_n = 0$; (3) normalized laplacian matrix $\mathcal{L} = I - D^{-\frac{1}{2}} A D^{-\frac{1}{2}}$ with spectrum $2 \ge \nu_1 \ge \nu_2 \ge \cdots \ge \nu_n = 0$.


\section{Spectral Representation of Propagation Graphs.}
\label{sec:encoding}

Existing topological features provide intuitive descriptions of propagation trees, but they are often introduced in an ad hoc manner and do not offer a unified view of cascade structure. As a result, they may fail to capture more global aspects of propagation, such as structural robustness and diffusion-related behavior. In contrast, spectral graph theory provides a common language for characterizing graph structure through eigenvalues and eigenvectors. Even when two cascades appear similar under standard topological statistics, their ``spectral fingerprints'' can still reveal important differences in how information propagates, as illustrated in Appendix~\ref{app:spectral-example} (see Fig. \ref{fig:spectral-example}). Motivated by this perspective, we study propagation trees through spectral representation. Directly using the full spectrum, however, is inconvenient because the number of eigenvalues varies with graph size. We therefore organize the spectral representation of news propagation into five complementary categories, each corresponding to one aspect of propagation behavior, such as branching capacity, structural cohesion, and diffusion dynamics. For each category, we use spectral bounds to connect graph spectra with propagation-related structural properties. These bounds include both classical results from spectral graph theory and several new bounds proved in this work. Together, they provide a propagation-oriented and interpretable representation of propagation trees. We next introduce these five categories and their representative spectral bounds to provide an intuitive understanding of the proposed spectral representation.\vspace{1mm}

\paragraph{C1. Branching Capacity} reflects how broadly a piece of news can expand at the local level, and prior studies have shown that fake and real news differ noticeably in this aspect~\cite{vosoughi2018spread,zhou2019network}. From a spectral perspective, branching capacity is naturally related to the spectral radius \(\lambda_1\), which is closely connected to degree concentration through bounds on the maximum and mean degree of a graph. Based on this connection, we further derive \(\lambda_1\)-based bounds for mean branching, maximum branching at layer \(k\), and degree entropy. The proofs are provided in Appendix~\ref{app:bounds_details}:

\begin{proposition}[Mean Branching Bound]
\label{prop:branching-factor}
Let $T$ be a rooted tree with $n$ nodes and $|I(T)|$ internal nodes. Then, for mean branching $\bar b$, 
\[
\bar b \le \frac{n\lambda_1}{2|I(T)|}.
\]
\end{proposition}
See Appendix~\ref{app:proof_mean_branch} for the proof.

\begin{proposition}[Maximum Branching at Layer $k$ Bound]
\label{prop:max-breadth}
Let $T$ be a rooted tree, and $b_k$ be the number of nodes at level $k$. Then,
\[
b_k \le \lambda_1^2(\lambda_1^2-1)^{k-1}, \qquad k\ge 1.
\]
\end{proposition}
See Appendix~\ref{app:proof_max_branch} for the proof.

\begin{proposition}[Degree Entropy Bound]
\label{prop:degree-entropy}
Let \(G\) be a graph with adjacency matrix \(A\), spectral radius \(\lambda_1\), and
degree distribution \(\{q_j\}\), where \(q_j\) is the fraction of nodes of degree
\(j\). Define the degree entropy by
\[
H_d = -\sum_j q_j \log q_j.\vspace{-3mm}
\]
Then
\[
H_d \le \log(\lambda_1^2+1).
\]
\end{proposition}
See Appendix~\ref{app:proof_degree_entropy} for the proof.\vspace{1mm}

\paragraph{C2. Cascade Scale} describes the overall size of information diffusion, including how many users and interactions are involved during propagation~\cite{vosoughi2018spread,shu2020hierarchical}. In spectral terms, graph size is naturally reflected in the Laplacian spectrum, since \(\sum_i \mu_i = \mathrm{trace}(L) = 2(n-1)\). This connection leads to several size-related bounds, such as \(\mu_1 \le n\) and \(\mu_2 \le \left\lfloor \frac{n}{2} \right\rfloor\)~\cite{zhang2011laplacian}.\vspace{1mm}

\paragraph{C3. Structural Cohesion} focuses on how robust and well connected a propagation structure is. Prior studies suggest that fake news is often associated with \textit{echo chambers}, where densely connected clusters are only weakly linked to the rest of the network~\cite{zhou2019network}. A key spectral quantity here is the algebraic connectivity \(\mu_{n-1}\). When \(\mu_{n-1}\) is small, the graph is easier to separate, which indicates weaker cohesion. This motivates the use of quantities such as vertex connectivity~\cite{brouwer2011spectra}, which reflects structural robustness, and the Cheeger constant~\cite{chung1997spectral}, which captures bottlenecks and global constraints in the propagation graph.\vspace{1mm}

\paragraph{C4. Propagation Span} characterize how far information propagates and how spread out the cascade becomes. Empirical studies suggest that fake news tends to reach deeper and farther parts of the network~\cite{vosoughi2018spread,friggeri2014rumor}. From a spectral perspective, this aspect is closely related to the algebraic connectivity \(\mu_{n-1}\). A larger \(\mu_{n-1}\) usually indicates a more cohesive structure, and therefore shorter distances between nodes; conversely, a smaller \(\mu_{n-1}\) is often associated with larger propagation span. Based on this connection, we consider spectral bounds for global distance-related quantities, including the diameter~\cite{alon1985lambda1}, and structural virality~\cite{sivasubramanian2009average}.\vspace{1mm}

\paragraph{C5. Diffusion Dynamics} unlike the previous categories, which describe static structural properties of propagation trees, diffusion dynamics concerns how efficiently information flows through the structure. This aspect is naturally linked to the normalized Laplacian spectrum, especially \(\nu_{n-1}\), which reflects how quickly diffusion mixes across the graph. Larger values of \(\nu_{n-1}\) indicate faster mixing and more efficient information spread, while smaller values suggest stronger diffusion bottlenecks. We therefore use quantities such as random walk convergence and routing time to characterize the dynamic behavior of news propagation~\cite{chung1997spectral}. \vspace{1mm}

Table~\ref{tab:spectral_bounds} summarizes all bounds, for a total of 35, organized into five categories of propagation structure.

\section{Structural Optimization of Propagation Trees.}
\label{sec:structural_decode}

After encoding propagation graphs with spectral bounds and training a classifier, we can further study news propagation patterns through structural optimization. The main idea is to modify an input propagation graph in a controlled direction with the help of the trained classifier. By tracing this structural evolution, we aim to understand how fake and real news differ in their propagation structures.

To this end, we first propose a unified discrete structural optimization algorithm. Within the algorithm, we consider two objectives. (1) \textbf{score-guided objective}: we maximize or minimize the predicted probability of being fake for a given graph, to observe how the structure changes when the graph becomes more fake-like or more real-like. (2) \textbf{Bound-guided objective}: we enforce a monotonic increase or decrease in a specific spectral bound (i.e., we change the graph structure), and then examine how the predicted fake probability changes along this process.

\subsection{Discrete Structural Optimization.}
To obtain meaningful evolution paths for propagation graphs, we need to control extraneous variables, such as the number of nodes, while also keeping each modification sufficiently fine-grained. We therefore adopt the following principles: 
\begin{enumerate}
    \item \textbf{Size control:} We fix the number of nodes so that the observed changes are purely structural, which avoids size-related bias in the classifier's prediction. 
    
    \item \textbf{Structural atomicity:} We define leaf node migration as the minimal unit of structural change. This yields a fine-grained evolution trajectory and supports first-order approximation. Moreover, any size-preserving structural change can be viewed as a sequence of leaf node migrations.
\end{enumerate} 
Based on these two principles, we present the abstract pipeline of our discrete structural optimization in Algorithm~\ref{alg:discrete_struc_optim}, where different objective functions can be specified in line 1. We then consider two objectives: (1) optimizing the current fake news prediction score in Section~\ref{sec:score-guided}; and (2) optimizing a specific propagation structure in Section~\ref{sec:bound-guided}.

\begin{algorithm}[tb]
\small
\caption{Discrete Structural Optimization}
\label{alg:discrete_struc_optim}
\begin{algorithmic}[1] 

\REQUIRE Initial tree $G_{0}$, Iterations $T$
\ENSURE Evolved tree sequence $\mathcal{H}=\{G_0, G_1, \dots, G_k\}$

\STATE Define objective function: $\Phi(G)$
\STATE Initialize: $\mathcal{H} := \{G_0\}$

\FOR{$t=0$ \TO $T$}
    \STATE $\mathcal{N}(G_{t}) := \{ \text{all possible leaf migrations on } G_{t} \}$
    
    \STATE $\delta^* := \operatorname{argmax}_{\delta \in \mathcal{N}(G_{t})} \Phi(G_t \oplus \delta)$
    
    \STATE $G_{t+1} := G_t \oplus \delta^*$
    \STATE $\mathcal{H} := \mathcal{H} \cup \{G_{t+1}\}$
\ENDFOR

\RETURN $\mathcal{H}$

\end{algorithmic}
\end{algorithm}

\subsection{Score-guided Optimization.}
\label{sec:score-guided}

We aim to decode the knowledge learned by the classifier into discrete graph topology. More specifically, given a fixed number of nodes, we ask what kind of propagation structure the classifier considers more characteristic of fake news. The full pseudo code is shown in Algorithm~\ref{alg:classifier_optim}. In line 1, we use the prediction score as the objective function. We use $d$ to control the direction of evolution, where $d=1$ ($d=-1$) drives the graph toward a more fake-like (real-like) structure. During the evolution process (lines 4--20), we first enumerate all possible leaf node migrations $\mathcal{P}$ (lines 5--13). Each candidate migration consists of a leaf node $v$, its current parent $p_{\mathrm{old}}$, and a new parent $p_{\mathrm{new}}$. For each candidate, we input the resulting graph into the classifier, compute the prediction score, and select the migration that gives the best objective value as the action at step $t$ (lines 8--12). After each action, we check whether the objective value $S$ has converged using a threshold $\tau$ (lines 14--19). After optimization, we obtain an evolved tree sequence $\mathcal{H}$ that reveals intuitive structural patterns associated with fake or real news propagation, and also provides an interpretable view of the classifier's behavior.

\begin{algorithm}[bhtp]
\caption{Score-guided Optimization}
\small
\label{alg:classifier_optim}
\begin{algorithmic}[1]
    \REQUIRE Initial tree $G_{init}$, Model $\mathcal{M}$, Iterations $T$, Direction $d \in \{1, -1\}$, Convergence threshold $\tau$
    \ENSURE Evolved  sequence $\mathcal{H}=\{(G_0, S_0), \dots, (G_k, S_k)\}$

    \STATE Define objective: $\Phi(G) = d \cdot \text{Prob}_{\mathcal{M}}(\text{fake} \mid G)$
    \STATE $G_{cur} \gets G_{init}$\;
    \STATE $\mathcal{H} \gets \{(G_{cur}, \Phi(G_{cur}))\}$
    
    \FOR{$t = 1$ \TO $T$}
        \STATE $\mathcal{N}(G_{cur}) \gets \{ \text{all valid leaf-node migrations } \delta=(v, p_{old}, p_{new}) \}$
        
        \STATE $G^*, S^* \gets \text{None}, \text{None}$
        
        \FORALL{candidate move $\delta \in \mathcal{N}(G_{cur})$}
            \STATE $G' \gets G_{cur} \oplus \delta$\;
            \STATE $S' \gets \Phi(G')$
            
            \IF{$S' > S^*$}
                \STATE $G^*, S^* \gets G', S'$
            \ENDIF
        \ENDFOR
        
        \IF{$(S^* - S_{cur}) > \tau$}
            \STATE $G_{cur}, S_{cur} \gets G^*, S^*$
            \STATE $\mathcal{H} \gets \mathcal{H} \cup \{(G_{cur}, S_{cur})\}$
        
        \ELSE  
            \STATE \textbf{break}
        \ENDIF
    \ENDFOR
    \RETURN $\mathcal{H}$
\end{algorithmic}
\end{algorithm}

\begin{algorithm}[t]
\caption{Bound-guided Optimization}
\label{alg:bound_optim}
\small
\begin{algorithmic}[1]
    \REQUIRE Initial tree $G_{0}$, Max iterations $T$, Direction $d \in \{1, -1\}$, Convergence threshold $\tau$
    \ENSURE Evolution sequence $\mathcal{H}=\{G_0, G_1, \dots, G_k\}$
    
    \STATE Define objective: $\Phi(G) = d \cdot \mathcal{B}(\mu_1, \dots, \mu_n)$ 
    \STATE $G_{cur} \gets G_0, \quad \mathcal{H} \gets \{G_0\}$
    
    \FOR{$t = 0$ \TO $T$}
        \STATE Compute all eigenvalues $\{\mu_i\}$ and eigenvectors $\{\mathbf{u}_i\}$ for $G_{cur}$
        
        \STATE $\mathcal{N}(G_{cur}) \gets \{ \text{all valid leaf-node migrations } \delta=(v, p_{old}, p_{new}) \}$
        
        \FORALL{candidate move $\delta \in \mathcal{N}(G_{cur})$}
            \STATE $\Delta \mu_i(\delta) \approx (u_{i}[v] - u_{i}[p_{new}])^2 - (u_{i,v} - u_{i}[p_{old}])^2, \quad \forall i$
            
            \STATE $\tilde{\mu}_i \gets \mu_i + \Delta \mu_i(\delta)$ 
            
            \STATE $\tilde{\Phi}(\delta) \gets d \cdot \mathcal{B}(\tilde{\mu}_1, \dots, \tilde{\mu}_n)$ 
        \ENDFOR
        
        \STATE $\delta^* = \arg\max_{\delta \in \mathcal{N}(G_{cur})} \left[ \tilde{\Phi}(\delta) \right]$
        
        \IF{$\tilde{\Phi}(\delta^*) - \Phi(G_{cur}) > \tau$}
            \STATE $G_{cur} \gets G_{cur} \oplus \delta^*$ 
            \STATE$\mathcal{H} \gets \mathcal{H} \cup \{G_{cur}\}$
        \ELSE
            \STATE \textbf{break}
        \ENDIF
    \ENDFOR
    \RETURN $\mathcal{H}$
    
\end{algorithmic}
\end{algorithm}

\subsection{Bound-guided Optimization.}
\label{sec:bound-guided}

In Section~\ref{sec:score-guided}, we optimized propagation structures indirectly through the classifier score, and then infer the global propagation patterns of fake and real news. In this section, we instead optimize specific structural bounds directly and examine how the evolution of these structural properties affects the fake news prediction. The full pseudo code is presented in Algorithm~\ref{alg:bound_optim}.

To optimize a target bound through leaf-node migration, we must evaluate how each candidate migration changes that bound. A direct approach is to compute the eigenvalues of every resulting graph exactly, and then use them to calculate the target bound. However, the number of candidate migrations in one step is usually large, ranging from $n-1$ to $(n-1)^2$, where $n$ is the number of nodes. For a path graph, there are $n-1$ possible migrations, while for a star graph, there are $(n-1)^2$. As a result, performing eigendecomposition for all candidates at every step is computationally expensive. To balance computational cost and estimation of accuracy, we propose a first-order approximation method to estimate the eigenvalue change caused by each leaf-node migration, as stated in Proposition~\ref{prop:estimate_mu}.

\begin{figure*}[t]
    \centering
    \begin{subfigure}[t]{0.32\textwidth}
        \centering
        \includegraphics[width=\linewidth]{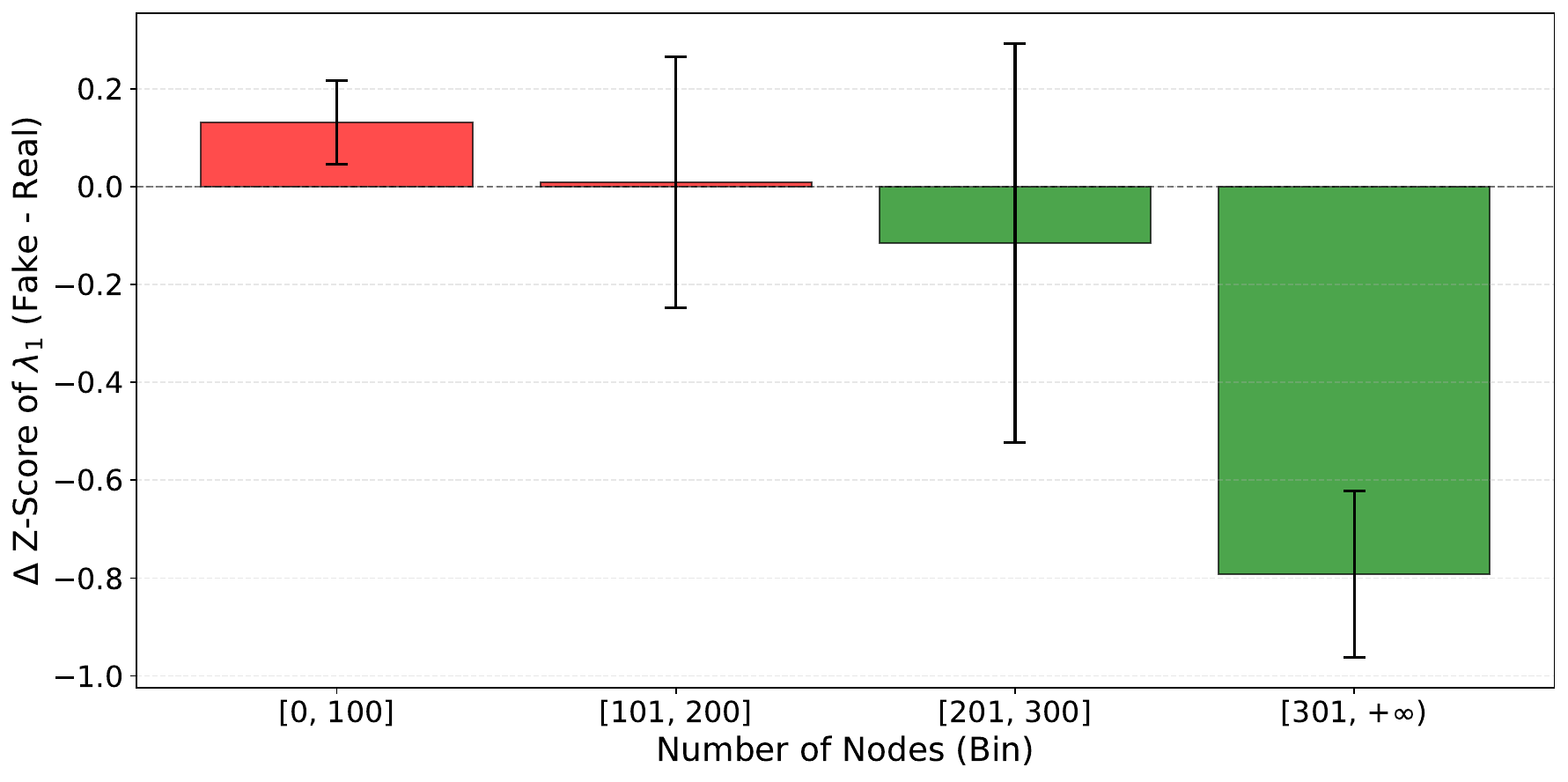}
        \caption{$\Delta z$-score of $\lambda_1$}
        \label{fig:lambda1}
    \end{subfigure}
    \hfill
    \begin{subfigure}[t]{0.32\textwidth}
        \centering
        \includegraphics[width=\linewidth]{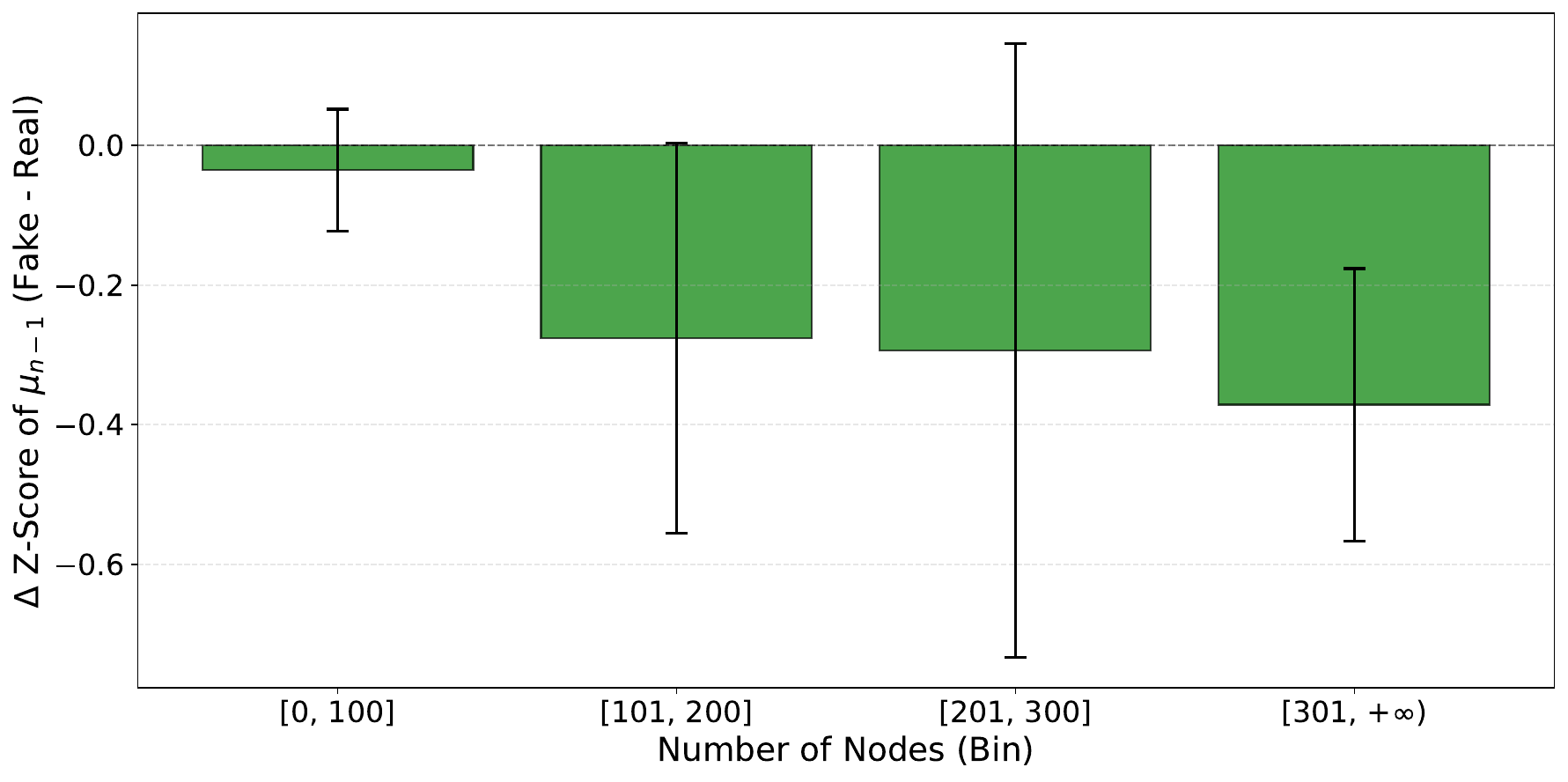}
        \caption{$\Delta z$-score of $\mu_{n-1}$}
        \label{fig:mu_n-1}
    \end{subfigure}
    \hfill
    \begin{subfigure}[t]{0.32\textwidth}
        \centering
        \includegraphics[width=\linewidth]{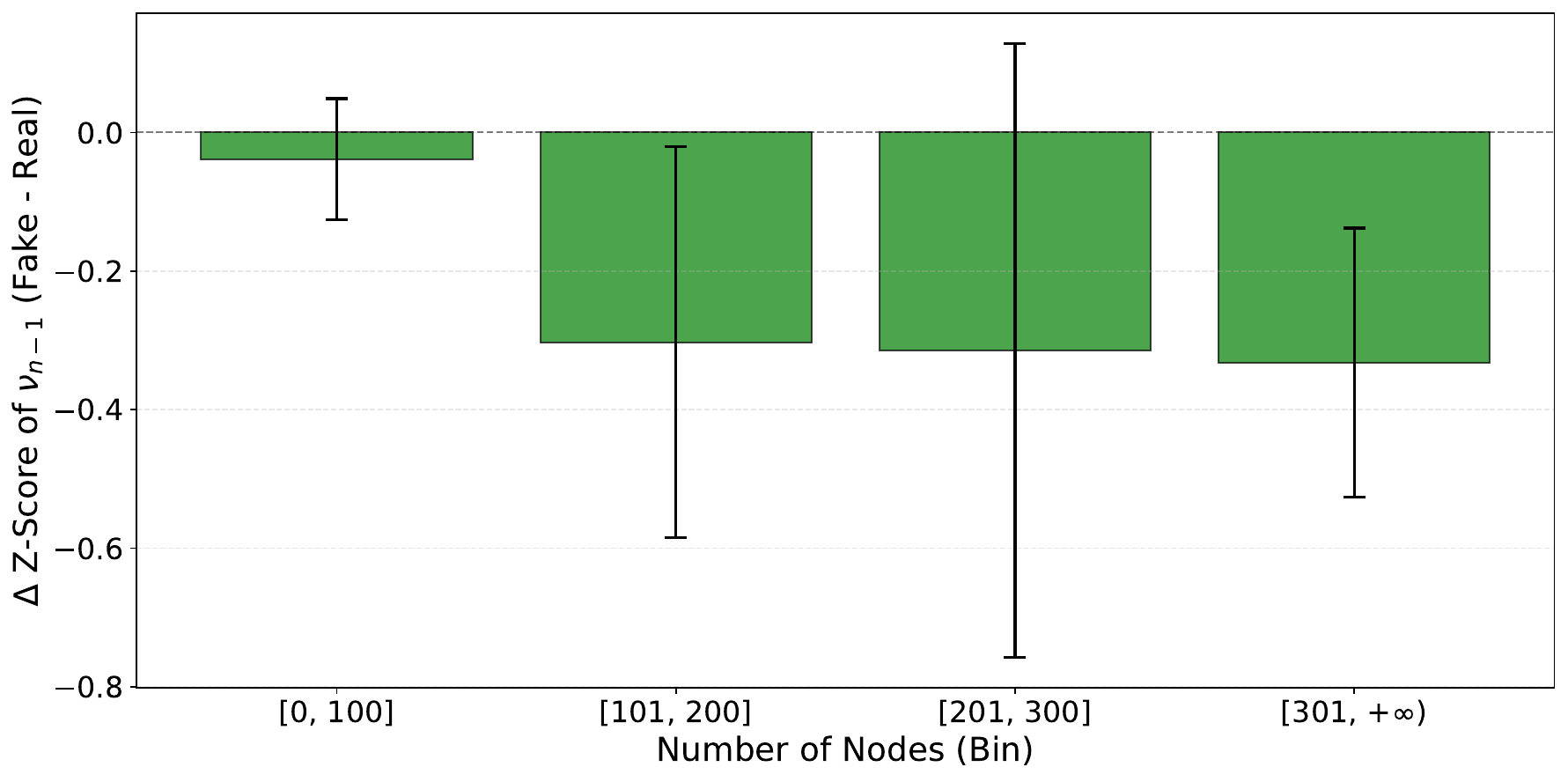}
        \caption{$\Delta z$-score of $\nu_{n-1}$}
        \label{fig:nu_n-1}
    \end{subfigure}\vspace{-3mm}
    \caption{\small Mean $z$-score differences across node-count bins for three eigenvalues, $\lambda_1$, $\mu_{n-1}$, and $\nu_{n-1}$, on the Weibo22 dataset. Red indicates that eigenvalues of fake news are larger than real news, while green indicates the opposite.}
    \label{fig:three_plots}
    \vspace{-4mm}
\end{figure*}

\begin{proposition}
\label{prop:estimate_mu}
Let $T$ be a tree and $u_i$ be its $i$-th orthonormal Laplacian eigenvector. When we move a leaf node $v$ from $p_{\text{old}}$ to $p_{\text{new}}$, we can estimate the change of each eigenvalue $\mu_i$ by
\begin{equation}
    \Delta \mu_i \approx 
    (u_i[v] - u_i[p_{\text{new}}])^2 
    - (u_i[v] - u_i[p_{\text{old}}])^2,
\label{eq:approx_delta_mu}
\end{equation}
\end{proposition}

\begin{proof}
Let $L u = \mu u$. After a small perturbation $\Delta L$, we have
\begin{align}
    (L + \Delta L)(u + \Delta u) = (\mu + \Delta \mu)(u + \Delta u).
\label{eq:perturb_L}
\end{align}
Ignoring higher-order terms, we obtain
\begin{align}
    L \Delta u + \Delta L u \approx \mu \Delta u + \Delta \mu u.
\end{align}
Left-multiplying by $u^T$ yields
\begin{align}
    u^T L \Delta u + u^T \Delta L u \approx \mu u^T \Delta u + \Delta \mu u^T u.
\end{align}
As $L$ is symmetric, $u^T L = \mu u^T$, and $u^T u = 1$, so
\begin{equation}
    \Delta \mu \approx u^T \Delta L u.
    \label{eq:approxmu}
\end{equation}
For adding new edge $(r,c)$, the perturbation matrix can be written as
\begin{equation}
    \Delta L = (e_r - e_c)(e_r - e_c)^T,
    \label{eq:deltaL}
\end{equation}
where $e_r$ is the standard basis vector.

Substituting Eq.~\ref{eq:deltaL} into Eq.~\ref{eq:approxmu}, we obtain
\begin{equation}
\label{eq:delta_mu_add}
    \Delta \mu_{\text{add}(r,c)} \approx (u_r - u_c)^2.
\end{equation}
Similarly, for deleting an edge $(r,c)$, we have
\begin{equation}
\label{eq:delta_mu_del}
    \Delta \mu_{\text{delete}(r,c)} \approx - (u_r - u_c)^2.
\end{equation}

By adding Eq~\ref{eq:delta_mu_add} and Eq~\ref{eq:delta_mu_del}, we can get Eq~\ref{eq:approx_delta_mu}.

\end{proof}
The same derivation can be applied to other types of eigenvalues, such as $\lambda_i$ and $\nu_i$. In particular, we obtain
\begin{equation}
\Delta \lambda_i \approx 
    2u_i[v]\bigl(u_i[p_{\text{new}}] - u_i[p_{\text{old}}]\bigr).
\label{eq:approx_delta_lambda}
\end{equation}

As the perturbation induced by a leaf-node migration is small, this first-order approximation works well in practice. The accuracy also tends to improve as the graph grows. For further analysis see Appendix~\ref{app:analysis_approx}.\vspace{1mm}

\noindent\textbf{Time Complexity.}
Let \(n\) be the number of nodes in the current propagation graph, and \(M\) be the number of valid leaf-node migrations in one iteration. As \(M\) ranges from \(O(n)\) to \(O(n^2)\), corresponding to chain and star graphs in the extreme cases, the cost of exact bound-guided optimization can be high. Without approximation, we need to recompute the eigendecomposition for every candidate migration, which gives a per-iteration time complexity of \(O(Mn^3)\). This ranges from \(O(n^4)\) to \(O(n^5)\). With first-order approximation, we perform eigendecomposition only once for the current graph, costing \(O(n^3)\), and then evaluate each candidate migration in \(O(n)\) time by estimating the change of eigenvalues. The resulting per-iteration complexity is \(O(n^3 + Mn)\), dominated by \(O(n^3)\). So, the approximation reduces per-iteration cost from \(O(Mn^3)\) to \(O(n^3)\).

\section{Experiments.}
\label{sec:exp}
We conduct a series of experiments to evaluate the proposed framework. First, we perform exploratory analyses of eigenvalues and spectral bounds in the news propagation setting. This allows us to examine whether fake and real news exhibit distinct structural differences, and to assess the tightness of the bounds introduced in Section~\ref{sec:encoding}. Second, we use these bounds as graph encodings for the downstream fake news classification task, in order to evaluate their effectiveness. Finally, based on the trained classifier, we apply the discrete structural optimization algorithm proposed in Section~\ref{sec:structural_decode}. This helps us obtain an intuitive understanding of the differences between fake and real news propagation patterns. More details about datasets and settings are provided in Appendix~\ref{app:exp_details}.

\begin{table*}[!t]
\centering
\caption{\small Comparison of spectral bounds for estimating structural properties on the \texttt{Weibo22} dataset. Entries marked with $^*$ are statistically significant at $p<0.05$. The bound formulas in the upper part of the table directly correspond to the meanings of the associated structural features, whereas those in the lower part do not. In the lower part, we further examine the relationship between the structure virality bound and the actual diameter, as well as the relationship between the max-breadth bound and the actual number of leaf nodes.}\vspace{-2mm}
\label{tab:tightness}
\resizebox{\textwidth}{!}{%
\scriptsize
\begin{tabular}{llllll}
\toprule
Bound Equation & Structural Feature & Rel. Error ($\epsilon$) & Spearman ($\rho$) & Kendall ($\tau$) & Pearson ($r$) \\
\midrule
$\lambda_1$ & max degree & $0.69 \pm 0.23$ & $0.999^*$ & $0.987^*$ & $0.905^*$ \\
\# distinct $\mu_i - 1$ & diameter & $5.29 \pm 8.1$ & $0.867^*$ & $0.718^*$ & $0.677^*$ \\
$\mu_1$ & num nodes & $0.18 \pm 0.20$ & $0.989^*$ & $0.928^*$ & $0.934^*$ \\
$\lambda_1^2(\lambda_1^2-1)^0$ (k=$1$) & max-breadth & $0.08 \pm 0.21$ & $0.998^*$ & $0.981^*$ & $0.998^*$ \\ 

\midrule

$\frac{2}{n-1} \sum 1/\mu_i$ & diameter & $0.43 \pm 0.20$ & $0.930^*$ & $0.799^*$ & $0.899^*$ \\

$\lambda_1^2(\lambda_1^2-1)^0$ (k=$1$) & \# leaf node & $0.14 \pm 0.22$ & $0.996^*$ & $0.959^*$ & $0.955^*$ \\

\bottomrule
\end{tabular}%
}
\vspace{-4mm}
\end{table*}

\subsection{Analysis of Eigenvalue Distribution.}
We compute the eigenvalues of every propagation graph in the \texttt{Weibo22} dataset. For fake and real news separately, we then standardize each eigenvalue by its $z$-score and compute the mean $\Delta z$-score, defined as the fake-news mean minus the real-news mean, for each eigenvalue. Figure~\ref{fig:three_plots} shows the distributions of $\lambda_1$, $\mu_{n-1}$, and $\nu_{n-1}$.

(1) Figure~\ref{fig:lambda1} shows that, for small propagation graphs, fake news usually has a larger $\lambda_1$. This suggests the presence of more highly connected hub nodes and a stronger ability for global diffusion. However, as the graph size increases, the $\lambda_1$ of real news becomes larger than that of fake news. This indicates that we \textit{cannot simply conclude that fake news spreads more broadly without considering graph size and related factors}.

(2) Figure~\ref{fig:mu_n-1} shows that real news usually has a larger $\mu_{n-1}$, which indicates a smaller graph diameter. A smaller diameter often corresponds to a shallower propagation depth. This is consistent with the finding that \textit{fake news tends to spread deeper}~\cite{vosoughi2018spread,zhou2019network}.

(3) Figure~\ref{fig:nu_n-1} shows that fake news often has a smaller $\nu_{n-1}$. This suggests a smaller \textit{Cheeger constant}, meaning that \textit{fake news propagation graphs tend to be less well connected, show clearer community separation, and are more prone to fragmentation}.

Overall, these observations suggest that eigenvalues can help distinguish fake news from real news. They also provide a certain degree of structural interpretability.

\subsection{Analysis of Bounds Tightness.}

We further examine whether the spectral bounds introduced in Section~\ref{sec:encoding} can faithfully reflect true graph properties in the news propagation setting. We evaluate the quality of these bounds from two perspectives: (1) \textbf{Tightness}, measuring how close a bound is to the corresponding graph property. We use the relative error $\epsilon_i = \frac{|q_i - b_i|}{q_i}$,
where $q_i$ denotes the true quantity and $b_i$ denotes its corresponding bound. A smaller $\epsilon_i$ shows a tighter bound. We compute the relative error for each graph in our data and report the mean value. (2) \textbf{Consistency}, measuring whether a bound preserves the relationship with the true property even when the numerical gap is large. We evaluate consistency from three aspects: (i) linear correlation, measured by Pearson's $r$~\cite{pearson1895vii}; (ii) monotonic correlation, measured by Spearman's $\rho$~\cite{spearman1961proof}; and (iii) pairwise agreement, measured by Kendall's $\tau$~\cite{kendall1938new}. 
The results in Table~\ref{tab:tightness} show two main findings. First, many spectral bounds are useful proxies for their corresponding graph properties. Some are reasonably tight with relative errors below $0.70$, while others remain highly correlated with the true quantities even when the numerical gap is larger. This indicates that spectral bounds can still provide reliable structural signals for characterizing news propagation. Second, cross-property comparisons reveal structural characteristics of dataset itself. For example, the structure virality bound is even closer to the true diameter than the diameter bound itself, and the max-breadth bound is close to the number of leaf nodes. Together, these observations suggest that propagation in \texttt{Weibo22} tends to be relatively star-like.

\begin{table}[!thbp]
\centering
\caption{\small Comparing classification performance on \texttt{Weibo22} and \texttt{Twitter16}. The best and second-best results are in bold and italics, respectively. Ties are broken by smaller standard deviation. Categories are abbreviated.}\vspace{-2mm}
\label{tab:classification_performance}
\resizebox{\linewidth}{!}{
\scriptsize
\setlength{\tabcolsep}{4pt}
\begin{tabular}{lcccc}
\toprule
& \multicolumn{2}{c}{\texttt{Weibo22}} & \multicolumn{2}{c}{\texttt{Twitter16}} \\
\cmidrule(lr){2-3} \cmidrule(lr){4-5}
Method & ACC & F1 & ACC & F1 \\
\midrule
Random & 47.92 & 49.16 & 22.98 & 22.94 \\
Structural Features & \textbf{59.65} & 56.52 & 33.03 & 30.79 \\
GCN & 48.86 & 47.77 & 27.42 & 27.35 \\
Bound Features & 59.40 & \textit{57.65} & \textbf{33.55} & \textbf{32.20} \\
\midrule
- \textit{C1. Branching} & \textit{59.47} & \textbf{57.96} & 33.16 & 31.51 \\
- \textit{C2. Scale} & 59.29 & 57.70 & 33.16 & 31.11 \\
- \textit{C3. Cohesion} & 59.25 & 57.70 & 33.55 & \textit{31.83} \\
- \textit{C4. Span} & 57.88 & 55.64 & 33.55 & 31.69 \\
- \textit{C5. Diffusion} & 58.31 & 54.53 & \textit{33.55} & 31.77 \\
\bottomrule
\end{tabular}
}
\vspace{-5mm}
\end{table}

\begin{figure*}[th]
    \includegraphics[width=\linewidth]{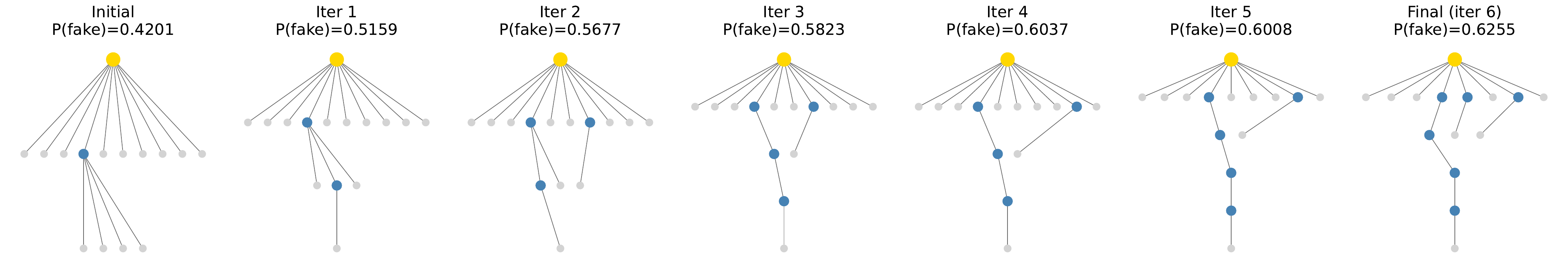}
    \vspace{-8mm}
    \caption{\small Propagation graph evolution toward a more fake news like structure under score-guided optimization. Yellow, blue, and gray nodes denote the root, internal, and leaf nodes, respectively. The title of each graph shows the current iteration $t$ and the classifier prediction score.}
    \label{fig:more_fake}
    \vspace{-1mm}
\end{figure*}

\begin{figure*}[th]
    \includegraphics[width=\linewidth]{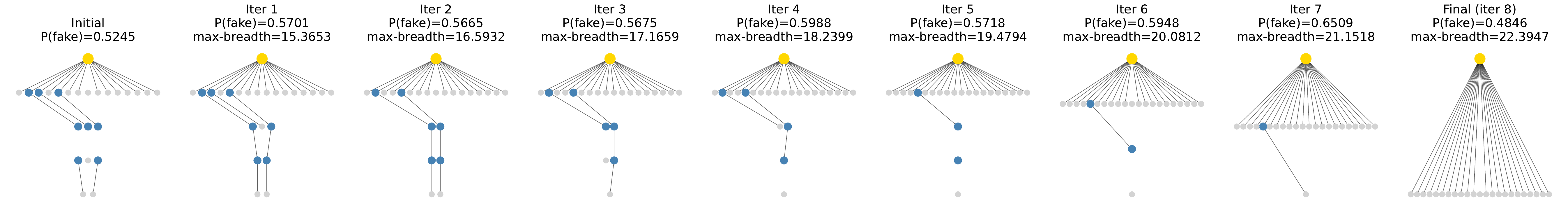}
    \vspace{-7mm}
    \caption{\small Propagation graph evolution toward higher max-breadth under bounds-guided optimization. Yellow, blue, and gray nodes denote the root, internal, and leaf nodes, respectively. The title of each graph shows the current iteration $t$, the classifier prediction score, and bound value.}
    \label{fig:increase_max_breadth}
    \vspace{-4.5mm}
\end{figure*}

\subsection{Effectiveness of Bounds in Classification.}
\label{sec:bounds_classifier}

To evaluate the effectiveness of spectral bounds in the downstream fake news classification using only structural information, we compare bound-based features with other structure-only methods, including hand-crafted topological features and GCN~\cite{kipf2016semi}. More details about the baselines are provided in Appendix~\ref{app:classification_details}. The results in Table~\ref{tab:classification_performance} show that both bound-based features and topological features perform much better than GCN. Moreover, bound-based features achieve performance comparable to topological features, even though some bounds are not numerically close to their exact properties. This suggests that spectral bounds still capture highly informative structural signals for distinguishing fake news from real news. In addition, by ablating the bounds in each category, we find that different datasets show different preferences over structural aspects. For example, on \texttt{Weibo22}, fake and real news differ more clearly in \textit{C4. Propagation Span} and \textit{C5. Diffusion Dynamics}, which is consistent with the distributional patterns in Figure~\ref{fig:three_plots}. In contrast, on \texttt{Twitter16}, the differences are more pronounced in \textit{C1. Branching Capacity} and \textit{C2. Cascade Scale}. This suggests that the propagation patterns of fake and real news may vary across events or platforms.

\subsection{Evolution of Propagation Structures.}
We apply algorithms proposed in Section~\ref{sec:structural_decode} with the trained classifier from Section~\ref{sec:bounds_classifier} to obtain Figure~\ref{fig:more_fake} and Figure~\ref{fig:increase_max_breadth}. The figures provide a deeper understanding about fake news propagation patterns. More cases can be found in Appendix~\ref{app:more_decode_cases}.\vspace{1mm}


\noindent \textbf{Evolution Toward ``More Fake.''} \textit{Although fake news propagation is often deeper and more unbalanced, increasing these factors alone does not always make a cascade more fake-like.} Figure~\ref{fig:more_fake} shows that the graph evolves from an initial depth of 2 to a final depth of 5, while becoming more unbalanced as a deep subtree gradually emerges. This suggests that fake-like propagation tends to be associated with deeper and more unbalanced structures. However, the effect is not monotonic. From iter 4 to iter 5, the depth increases but \(P(\text{fake})\) decreases. From iter 5 to iter 6, the algorithm does not further increase depth, but instead moves a first-layer node to the second layer, which raises \(P(\text{fake})\).

\vspace{1mm}

\noindent \textbf{Evolution Toward Higher Max-Breadth.} \textit{Although prior studies suggest that fake news typically spreads more broadly~\cite{vosoughi2018spread}, a broader propagation pattern does not necessarily make a cascade more fake-like.} Figure~\ref{fig:increase_max_breadth} shows that as the graph evolves to increase max-breadth, \(P(\text{fake})\) does not follow a stable monotonic trend. For example, the trees at iter 4 and iter 5 have the same depth, but the broader tree at iter 5 still receives a lower \(P(\text{fake})\). From iter 5 to iter 8, depth continues to decrease and breadth continues to increase, yet \(P(\text{fake})\) keeps rising. However, in the final iteration, \(P(\text{fake})\) drops sharply. This suggests that whether a propagation pattern appears more fake-like or more real-like depends on the overall structural configuration, rather than on any single factor alone.
\vspace{1mm}

\noindent Together, these evolutions suggest that \textit{fake-like propagation depends on a more specific structural configuration than simply being deeper, unbalanced, or broader.}

\section{Conclusion.}
We have presented, to our knowledge, the first spectral representation of news propagation graphs. By connecting the graph spectrum to structural properties, our representation captures branching, scale, cohesion, span, and diffusion within one framework rather than as separate handcrafted features. Our performance matches handcrafted features on structure-only classification and, through perturbation-based decoding, yields interpretable evolution trajectories showing that the fake-real distinction rests on a cascade's global configuration rather than any single structural trait. \textit{Spectral analysis thus offers a general, systematic, and interpretable approach for studying how information spreads, fake or otherwise.}\vspace{3mm}

\noindent \textbf{Acknowledgment.} This research was supported in part by the NSF under awards 2241070 and 1942929.

\bibliographystyle{siamplain}
\bibliography{references}

\appendix

\section{Example: Limits of Standard Topological Statistics.}
\label{app:spectral-example}

\begin{figure*}[t]
    \centering
    \includegraphics[width=0.85\textwidth]{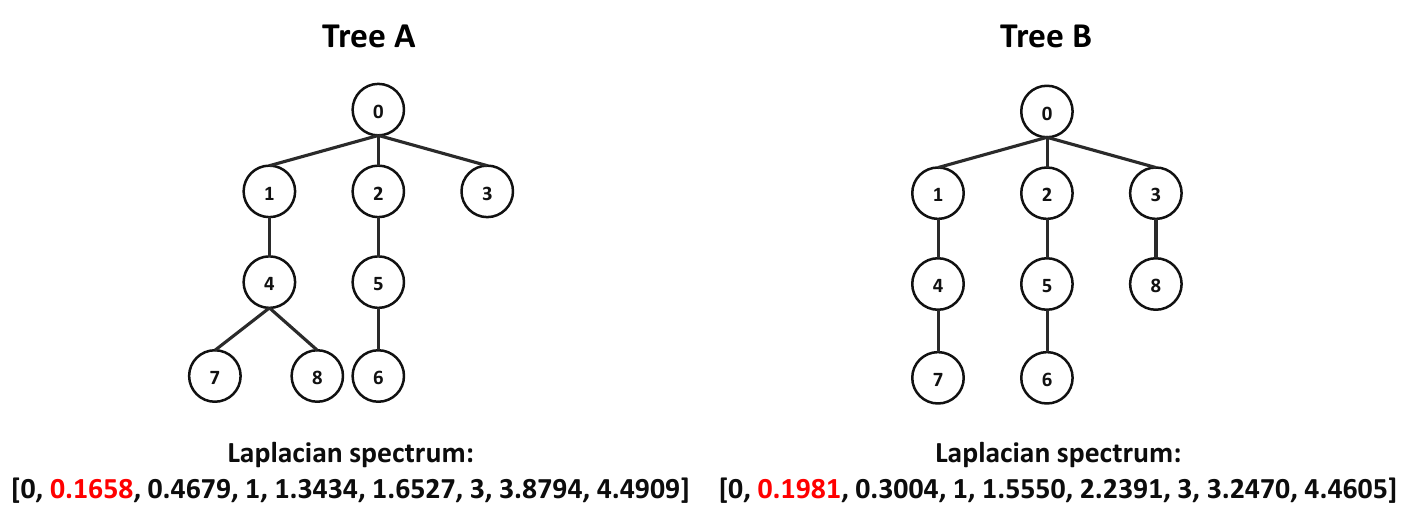}
    \caption{
    Two propagation trees with identical topological statistics but different Laplacian spectra. 
    Both trees have $n=9$, depth $3$, maximum out-degree $3$, maximum breadth $3$, structural virality $2.83$, and depth of maximum out-degree $0$. 
    However, their second smallest Laplacian eigenvalues differ: $\mu_{8}=0.1658$ for Tree A and $\mu_{8}=0.1981$ for Tree B, indicating different levels of algebraic connectivity and structural compactness.
    }
    \label{fig:spectral-example}
\end{figure*}

To illustrate why spectral representations can provide information beyond standard topological statistics, we consider two rooted propagation trees with nine nodes, shown in Figure~\ref{fig:spectral-example}. The two trees have the same number of nodes, depth, maximum out-degree, maximum breadth, structural virality, and depth of the maximum out-degree node. Therefore, a representation based only on these commonly used statistics would treat them as nearly indistinguishable.
However, their Laplacian spectra are different. In particular, the second smallest Laplacian eigenvalue is $0.1658$ for Tree A and $0.1981$ for Tree B. This eigenvalue, also known as \textit{algebraic connectivity}, reflects global connectivity and bottleneck structure. The larger value for Tree B suggests that it is more compact and balanced, while Tree A contains a more pronounced local branching imbalance. This example shows that \textit{even when cascades appear similar under a set of handcrafted topological features, their ``spectral fingerprints'' can reveal propagation-relevant structural differences}.

\section{Spectral Bounds Details.}
\label{app:bounds_details}

\begin{table*}[!thbp]
\centering
\footnotesize
\renewcommand{\arraystretch}{1.3}
\caption{Spectral bounds categorized by propagation-related structural properties.}

\begin{tabularx}{\textwidth}{
>{\raggedright\arraybackslash}p{4cm}
>{\raggedright\arraybackslash}p{6.5cm}
>{\raggedright\arraybackslash}X
}

\toprule
\textbf{Category} & \textbf{Meaning / Property} & \textbf{Bound Equation} \\
\midrule

\multirow{7}{*}{C1. Branching Capacity}

& Max degree / Mean degree
& $\bar d \le \lambda_1 \le \Delta$ \cite{brouwer2011spectra} \\

& Maximum sum of degrees of adjacent nodes
& $\mu_1 \le \max_{x\sim y}(d_x+d_y)$ \cite{brouwer2011spectra} \\

& Sum of top-$m$ degrees
& $1+\sum_{i=1}^{m} d_i \le \sum_{i=1}^{m} \mu_i$ \cite{brouwer2011spectra} \\

& Mean branching
& $\bar b \le \dfrac{n\lambda_1}{2|I(T)|}$ \textcolor{blue}{Prop \ref{prop:branching-factor}}\\

& Maximum branching at layer $k$
& $b_k \le \lambda_1^2(\lambda_1^2-1)^{k-1}$  \textcolor{blue}{Prop \ref{prop:max-breadth}}\\


& Degree entropy
& $H_d \le \log(\lambda_1^2+1)$ \textcolor{blue}{Prop \ref{prop:degree-entropy}} \\

\midrule

\multirow{6}{*}{C2. Cascade Scale}

& Number of edges
& $\dfrac{2e}{n} \le \lambda_1 \le \sqrt{2e}$ \cite{brouwer2011spectra} \\

& 
& $\lambda_1(\lambda_1+1) \le 2e$ \cite{brouwer2011spectra} \\

& 
& $\sum_{i=1}^{t} \mu_i \le e+\binom{t+1}{2}$ \cite{brouwer2011spectra} \\

& Edge / node relation
& $\lambda_1 \le \sqrt{2e-n+1}$ \cite{brouwer2011spectra} \\

& Number of nodes
& $\mu_2 \le \left\lfloor \frac{n}{2} \right\rfloor$ \cite{zhang2011laplacian} \\

& 
& $\mu_1 \le n$ \cite{zhang2011laplacian} \\

&
& $\sum_i \nu_i \le n$ \cite{chung1997spectral} \\

& 
& $\nu_{n-1} \le \frac{n}{n-1}$ \cite{chung1997spectral} \\

& 
& $\nu_1 \ge \frac{n}{n-1}$ \cite{chung1997spectral} \\

\midrule

\multirow{6}{*}{C3. Structural Cohesion}

& Separation
& $\dfrac{|X||Y|}{(n-|X|)(n-|Y|)} \le
\left(\dfrac{\mu_1-\mu_{n-1}}{\mu_1+\mu_{n-1}}\right)^2$ \cite{brouwer2011spectra}\\

&
& $\dfrac{|X||Y|}{n(n-|X|-|Y|)} \le
\dfrac{(\mu_1-\mu_{n-1})^2}{4\mu_1\mu_{n-1}}$ \cite{brouwer2011spectra} \\

& Vertex connectivity
& $\kappa(G) \ge \mu_{n-1}$ \cite{brouwer2011spectra} \\

& Cheeger constant
& $\dfrac{\nu_{n-1}}{2} \le h(G) \le \sqrt{2\nu_{n-1}}$ \cite{chung1997spectral}\\

&
& $h^{'}(H) \ge \frac{\mu_{n-1}}{2}$ \cite{naderi2026introduction}\\


& Independent number
& $\alpha(G)\le \dfrac{-n\lambda_n}{\lambda_1-\lambda_n}$ \cite{brouwer2011spectra} \\

& 
& $\alpha(G)\le
\min(\{\lambda_i\ge0\},  \{\lambda_i\le0\})$ \cite{naderi2026introduction} \\

& Chromatic number
& $1-\dfrac{\lambda_1}{\lambda_n}\le \chi(G)\le 1+\lambda_1$ \cite{brouwer2011spectra} \\

& Clique number
& $\dfrac{n}{n-\lambda_1} \le \omega(G) \le 1+\lambda_1$ \cite{naderi2026introduction} \\

\midrule

\multirow{7}{*}{C4. Propagation Span}

& Bandwidth
& $b = \left\lceil \dfrac{n\mu_{n-1}}{\mu_1} \right\rceil$ \cite{brouwer2011spectra} \\

& Structural virality
& $\dfrac{2}{n-1}\sum_{i=1}^{n-1}\dfrac{1}{\mu_i}$ \cite{sivasubramanian2009average} \\

& Diameter
& $\dfrac{4}{n\mu_{n-1}} \le D(G) \le
2\sqrt{\dfrac{2\Delta}{\mu_{n-1}}}\log_2 n$ \cite{alon1985lambda1} \\

& (for regular graph)
& $D(G) \le
\dfrac{\log(n-1)}{\log\left(\dfrac{\nu_1+\nu_{n-1}}{\nu_1-\nu_{n-1}}\right)}$  \cite{chung1997spectral} \\

& 
& $D(G) \le \text{distinct number of } \nu_i \text{ or } \mu_i - 1$ \cite{chung1997spectral, brouwer2011spectra} \\

&
& $m_T[0,1) \ge (d+1)/3$~\cite{guo2025laplacian} \\

\midrule

\multirow{3}{*}{C5. Diffusion Dynamics}


& Random walk convergence
& $O\!\left(\dfrac{\log n}{\nu_{n-1}}\right)$ \cite{chung1997spectral} \\

& Routing time
& $rt(G,\sigma)=O\!\left(\dfrac{\log^2 n}{\nu_{n-1}}\right)$ \cite{chung1997spectral} \\

& Conductance
& $\dfrac{\mathrm{vol}(\delta X)}{\mathrm{vol}(X)} \ge
\dfrac{2\nu_{n-1}}{\nu_1+\nu_{n-1}}$ \cite{chung1997spectral} \\

& (coefficient of one bound)
& $\frac{2\nu_{n-1}}{\nu_1 + \nu_{n-1}} (2- \frac{2\nu_{n-1}}{\nu_1 + \nu_{n-1}})$ \cite{chung1997spectral} \\

& Spectral moments
& $m_2,\; m_4$ \cite{jin2020spectral} \\

\bottomrule
\end{tabularx}

\label{tab:spectral_bounds}
\end{table*}












We list all bounds in Table~\ref{tab:spectral_bounds}. Here are our proofs for mean branching, maximum branching at layer $k$, and degree entropy:

\begin{lemma}~\cite{das2004some}
\label{lem:max-degree}
Let $T$ be a tree with $n$ nodes, Let $\Delta$ be the maximum degree of $T$, then we have 
\begin{equation}
\Delta \le \lambda_1^2
\end{equation}
\end{lemma}

    

\subsection{Proof of Proposition 4.1 (Mean Branching Bound)}
\label{app:proof_mean_branch}

\begin{proof}
Let $k_v$ be the number of children of an internal node $v$. Then
\[
\bar b=\frac{1}{|I(T)|}\sum_{v\in V_{\mathrm{int}}} k_v.
\]
Since every edge in a rooted tree contributes exactly once to a parent-child relation,
\[
\sum_{v\in V_{\mathrm{int}}} k_v = n-1.
\]
Thus
\[
\bar b=\frac{n-1}{|I(T)|}.
\]
On the other hand, the average degree is
\[
\bar d=\frac{2(n-1)}{n}.
\]
Hence
\[
\bar b=\frac{n}{2|I(T)|}\bar d.
\]
Using the standard bound $\bar d\le \lambda_1$~\cite{brouwer2011spectra}, we obtain
\[
\bar b \le \frac{n\lambda_1}{2|I(T)|}.
\]
\end{proof}

\subsection{Proof of Proposition 4.2 (Maximum Branching at Layer k)}
\label{app:proof_max_branch}

\begin{proof}
Let $\Delta$ be the maximum degree of $T$. The root has at most $\Delta$ children, so
\[
b_1\le \Delta.
\]
For $k\ge 2$, each node at level $k-1$ has at most $\Delta-1$ children, since one incident edge is used by its parent. Therefore,
\[
b_k \le \Delta(\Delta-1)^{k-1}.
\]

Now we leverage Lemma~\ref{lem:max-degree}. By substituting which, we obtain
\[
b_k \le \lambda_1^2(\lambda_1^2-1)^{k-1}.
\]
\end{proof}

\subsection{Proof of Proposition 4.3 (Degree Entropy Bound)}
\label{app:proof_degree_entropy}

\begin{proof}
Let \(K\) denote the number of distinct degree values that appear in \(G\). As
degrees are integers between \(0\) and the maximum degree \(\Delta\), we have
\[
K \le \Delta+1.
\]

For a discrete distribution supported on \(K\) values, entropy is maximized by the
uniform distribution. Hence
\[
H_d = -\sum_j q_j \log q_j \le \log K \le \log(\Delta+1).
\]

As we already know \(\Delta \le \lambda_1^2\) from Lemma~\ref{lem:max-degree}, thus
\[
H_d \le \log(\Delta+1) \le \log(\lambda_1^2+1).
\]
\end{proof}




\section{Computational Analysis of First-Order Approximation.}
\label{app:analysis_approx}

\begin{figure*}[thbp]
    \centering
    \begin{subfigure}[t]{0.32\textwidth}
        \centering
        \includegraphics[width=\linewidth]{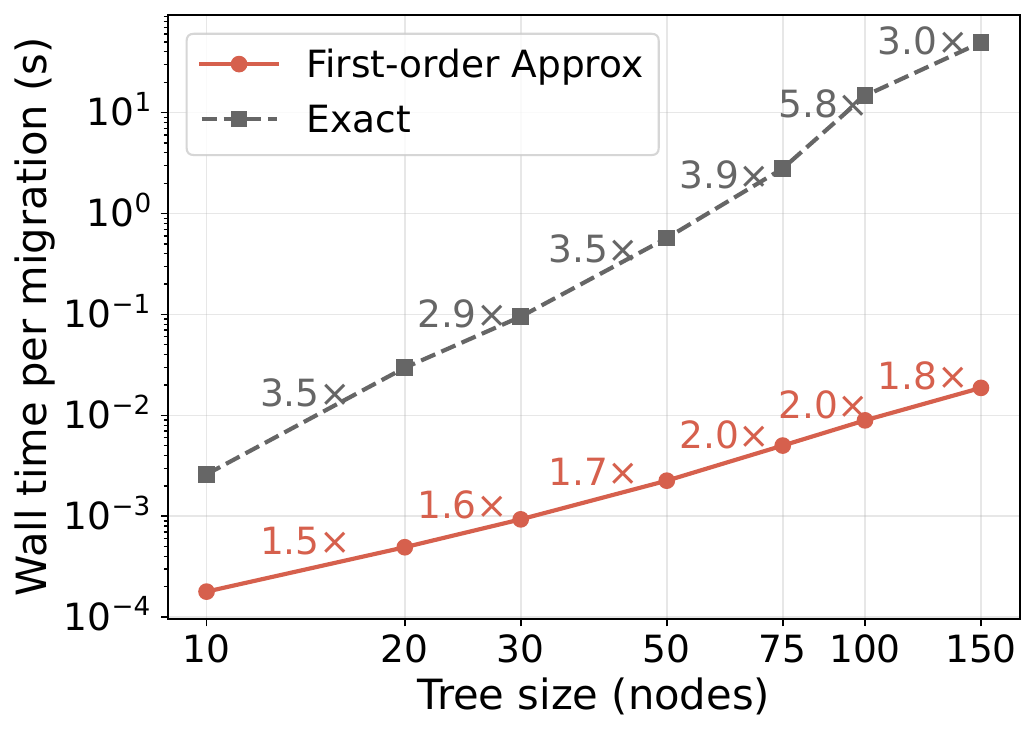}
        \caption{Wall time across tree size.}
        \label{fig:approx_wall_time}
    \end{subfigure}
    \hfill
    \begin{subfigure}[t]{0.32\textwidth}
        \centering
        \includegraphics[width=\linewidth]{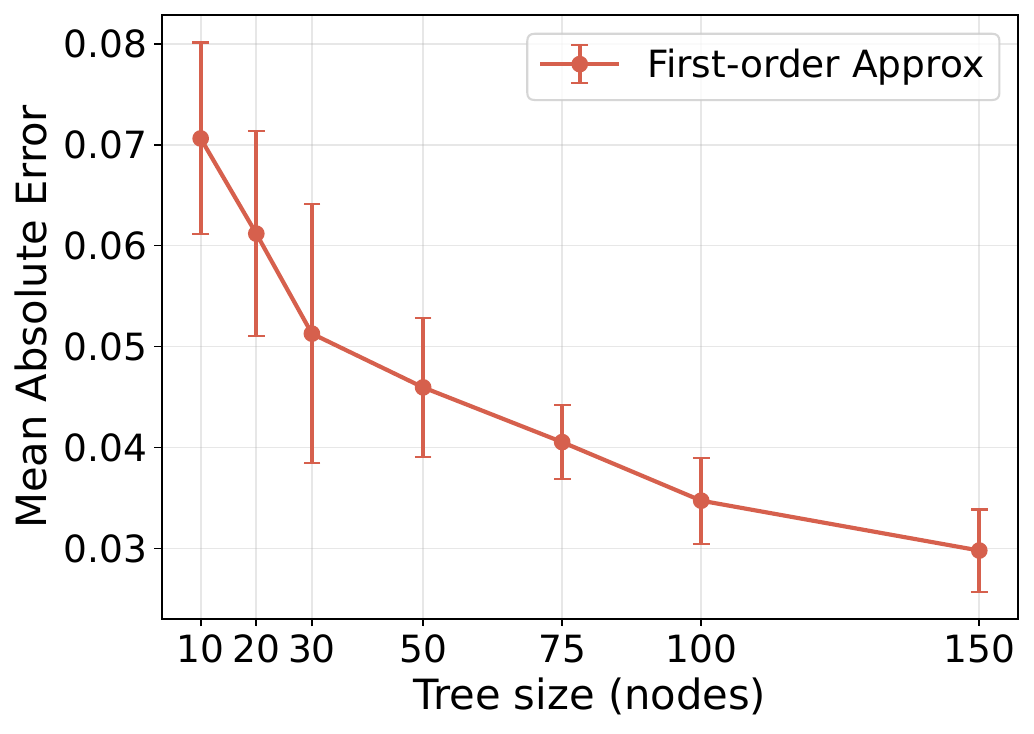}
        \caption{MAE across tree size.}
        \label{fig:approx_mae}
    \end{subfigure}
    \hfill
    \begin{subfigure}[t]{0.32\textwidth}
        \centering
        \includegraphics[width=\linewidth]{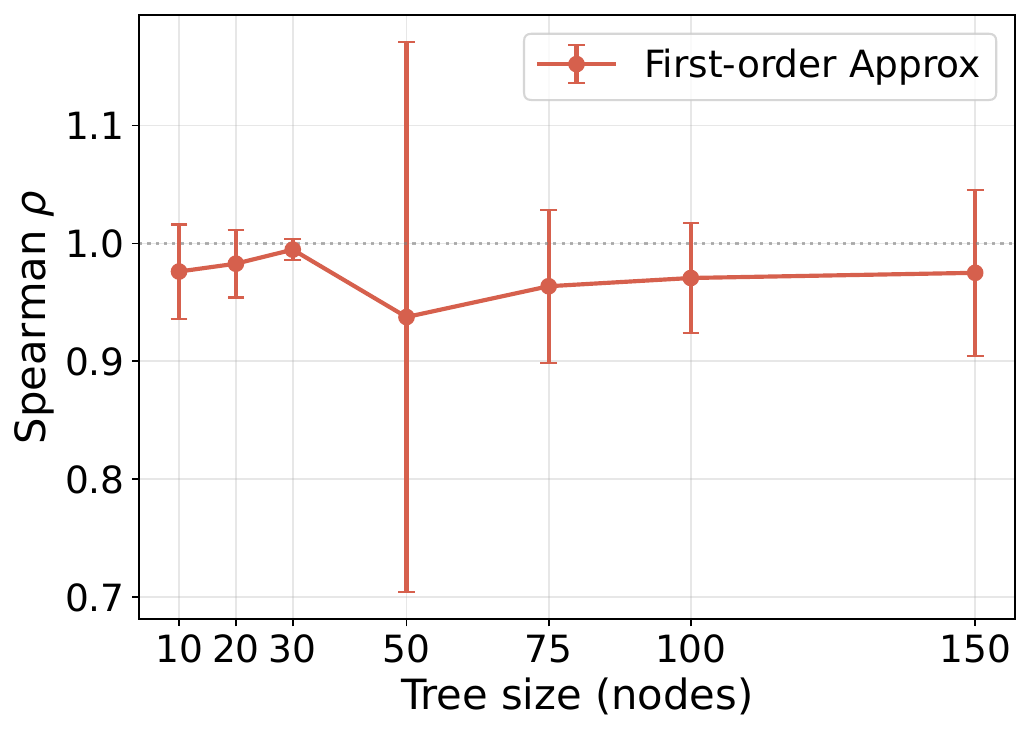}
        \caption{Spearman $\rho$ across tree size.}
        \label{fig:approx_spearman}
    \end{subfigure}
    
    \caption{Efficiency and accuracy of the first-order approximation on the Weibo22 dataset, using \(\lambda_1\) as the target quantity. In Figure~\ref{fig:approx_wall_time}, both axes are shown on a log scale, and the numbers annotated along the curves indicate the corresponding slopes in the log-log plot.}
    \label{fig:approx_three_plots}
\end{figure*}

To evaluate the efficiency and accuracy of the proposed first-order approximation, we compare it with the exact method, which recomputes eigenvalues by eigendecomposition on the Weibo22 dataset. Specifically, we consider a range of target tree sizes, \([10, 20, 30, 50, 75, 100, 150]\), and for each size randomly sample up to 20 trees whose node counts differ by at most 3 from the target. We then apply both methods to these trees to compute our target eigenvalue after per migration and compare them in terms of: (1) practical running time, measured by wall time per migration; and (2) approximation quality, measured by the mean absolute error (MAE) and Spearman's \(\rho\). We report the mean and standard deviation over all sampled trees. We use \(\lambda_1\), which is involved in the max-breadth bound~\ref{prop:max-breadth}, as the target quantity. The results are shown in Figure~\ref{fig:approx_three_plots}.

\noindent \textbf{Efficiency.}
Figure~\ref{fig:approx_wall_time} shows that the first-order approximation is consistently faster than the exact method. On the log-log scale, the approximation grows much more slowly with tree size. For small trees with 10 nodes, it is already about one order of magnitude faster. When the tree size reaches 150, the gap increases to roughly three orders of magnitude. This confirms that the approximation substantially reduces the computational cost of bound-guided optimization in practice.

\noindent \textbf{Accuracy.}
Figures~\ref{fig:approx_mae} and~\ref{fig:approx_spearman} show that the approximation remains accurate across all tree sizes. As the number of nodes increases, the MAE steadily decreases, while the average Spearman's \(\rho\) stays above 0.9. Notably, larger trees also induce substantially more candidate leaf-node migrations, typically growing toward \(O(n^2)\), which makes candidate ranking increasingly difficult. Even under this harder setting, the first-order approximation still preserves the ranking of candidate migrations well. This indicates that it remains reliable for guiding bound-guided optimization, especially on larger graphs.


\section{Experiments Details.}
\label{app:exp_details}
We detail the experimental setup, which includes datasets and baselines.

\subsection{Dataset Details.}
\label{app:dataset_details}
Weibo22~\cite{zhang2025rumor} is a recently collected Sina Weibo rumor detection dataset. The dataset covers events from November 2019 to March 2022, with more than half of the events related to the COVID-19
pandemic. 
Twitter16~\cite{ma2017detect} is a widely used Twitter rumor detection benchmark that contains source tweets and their propagation structures, with events labeled as non-rumors, false rumors, true rumors, or unverified rumors.
During experiments, we only use the graph with more than $3$ nodes and less than $10,000$ nodes. 
Statistics can be seen from Table~\ref{tab:dataset_stats}. 

\begin{table}[!thbp]
\centering
\caption{Statistics of Twitter16 and Weibo22.}
\label{tab:dataset_stats}
\begin{tabular}{lcc}
\toprule
Statistic & Twitter16 & Weibo22 \\
\midrule
\# Graphs & 766 & 2761 \\
\# Non-rumors & 205 & 1247 \\
\# False rumors & 180 & 1514 \\
\# True rumors & 196 & 0 \\
\# Unverified rumors & 185 & 0 \\
Avg \# Nodes & 24.2 & 228.3 \\
Med \# Nodes & 14 & 20 \\
Max \# Nodes & 250 & 9495 \\
\bottomrule
\end{tabular}
\end{table}



\subsection{Classification Experiment Details.}
\label{app:classification_details}
We use logistic regression (LR) as the classifier and apply 5-fold cross-validation to every model. In all experiments, we only use graph structural features, without incorporating any textual features. 
Their respective settings are as follows:
\begin{itemize}
    \item \textbf{Random}: we uniformly sample labels from $(0, 1)$ as prediction;
    \item \textbf{GCN}: we use random features as node features since we only use structure information. We further manually tune hyperparameters and select the best as final parameters. Finally, we set number of layers as $2$, hidden dimension as $64$, node feature dimension as $32$, learning rate as $0.001$, batch size as $32$, and total epoch as $100$;
    \item \textbf{Structural Features}: We use following structural features, and 32 of them are meaningful for trees:
\begin{enumerate}
    \item number of nodes;
    \item number of edges;
    \item depth;
    \item maximum of breadth;
    \item structural virality;
    \item maximum of out degree;
    \item depth of the maximum of out degree node;
    \item width entropy: $\sum_i p(w_i)\log(p(w_i))$, $p(w_i)=\frac{w_i}{\sum_i w_i}$, where $w_i$ is the width of depth $i$;
    \item leaf node ratio;
    \item average of depth;
    \item variance of depth;
    \item average of leaf depth;
    \item mean branching: $\frac{\sum_{i \in I(T)} d_i}{|I(T)|}$, where $|I(T)|$ is the number of internal nodes, $d_i$ is the degree of node $i$;
    \item variance of branching;
    \item sackin index: $S(T) = \sum_{v \in L(T)} d(v)$, i.e., the sum of the depths of all leaf nodes, where $T$ is a tree, $L(T)$ is a leaf node set of $T$, $d(v)$ is the depth of node $v$;
    \item colless index: $C(T) = \sum_{v \in I(T)} |L_v - R_v|$, where $L_v$ and $R_v$ represent number of leaf nodes in $v$'s left subtree and right subtree;
    \item number of internal node;
    \item degree entropy;
    \item degree gini: $G_d = \frac{\sum_{i=1}^{n}\sum_{j=1}^{n} |d_i - d_j|}{2n\sum_{i=1}^{n} d_i}$, which quantifies the inequality of the degree distribution. A larger value indicates a more uneven concentration of node degrees;
    \item diameter;
    \item radius;
    \item mean degree;
    \item maximum of degree;
    \item chromatic number: it is a constant number $2$ for trees;
    \item approximate independent number;
    \item maximum of degrees of a pair of connected nodes;
    \item sum of top 30\% degree;
    \item sum of top 60\% degree;
    \item $e + \binom{t+1}{2}$, where $t=0.3 * \lfloor n \rfloor$, which is an upper bound for sum of part laplacian eigenvalues (check bounds of number of edges in Table~\ref{tab:spectral_bounds});
    \item $e + \binom{t+1}{2}$, where $t=0.6 * \lfloor n \rfloor$, which is an upper bound for sum of part laplacian eigenvalues (check bounds of number of edges in Table~\ref{tab:spectral_bounds});
    \item bandwidth: $\operatorname{bw}(L)=\max_{L_{ij}\neq 0} |i-j| + 1$, where $L$ is the laplacian matrix of graph, which measures how far the nonzero entries of its matrix representation spread away from the main diagonal under a given node ordering;
    \item wiener index: the sum of the shortest-path distances between all unordered pairs of nodes;
    \item number of leaf nodes;
    \item number of spanning tree: its a constant number $1$ for trees;
    
\end{enumerate}
\end{itemize}


\section{More Structural Optimization Results.}
\label{app:more_decode_cases}

\begin{figure*}[!bhtp]
    \includegraphics[width=\linewidth]{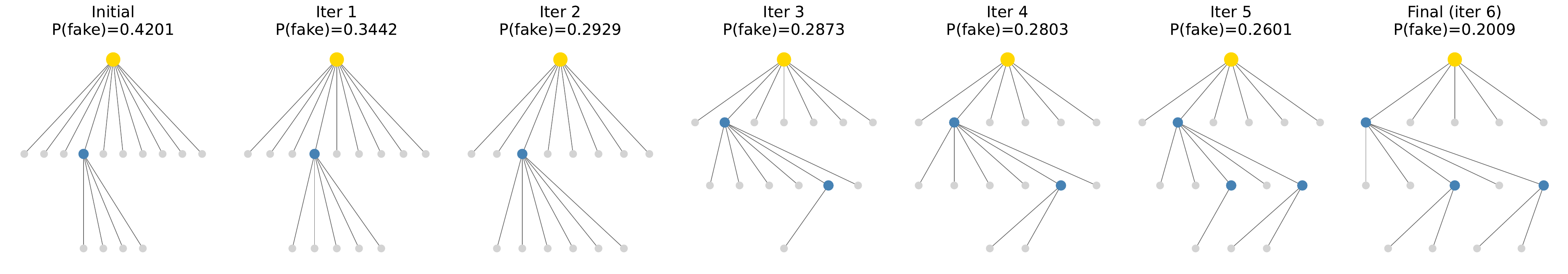}
    \caption{Propagation graph evolution toward a more real news like structure under classifier-guided optimization. Yellow, blue, and gray nodes denote the root, internal, and leaf nodes, respectively. The title of each graph shows the current iteration $t$ and the classifier prediction score.}
    \label{fig:more_real}
\end{figure*}

\begin{figure*}[!bhtp]
    \includegraphics[width=\linewidth]{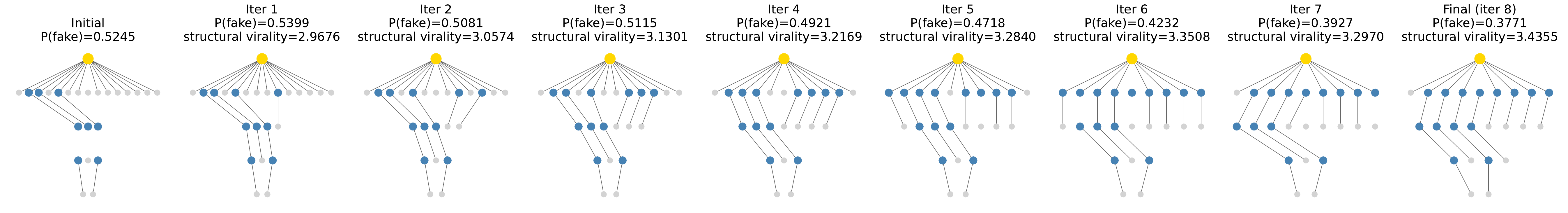}
    \caption{Propagation graph evolution toward higher structural virality under bounds-guided optimization. Yellow, blue, and gray nodes denote the root, internal, and leaf nodes, respectively. The title of each graph shows the current iteration $t$, the classifier prediction score, and bound value.}
    \label{fig:increase_sv}
\end{figure*}

In addition to the two cases presented in the main text, we provide more examples of the proposed discrete structural optimization algorithms to further illustrate the differences between fake and real news dissemination patterns.

\subsection{Evolution Towards More Real.}
We use the same initial tree as in Figure~\ref{fig:more_fake}, but reverse the optimization direction to make the propagation pattern more real-like. The resulting evolution is shown in Figure~\ref{fig:more_real}. By comparing it with the evolution toward a more fake-like structure, we can gain an intuitive understanding of the structural differences between fake and real news propagation.

We observe that, in Figure~\ref{fig:more_real}, the tree becomes more balanced when the blue node in layer 1 is treated as the root. At the same time, its maximum breadth decreases as the structure evolves toward a more real-like pattern.

\subsection{Evolution Towards Higher Structural Virality.}
We use the same initial tree as in Figure~\ref{fig:increase_max_breadth}, but change the target bound to the structural virality bound. The resulting evolution is shown in Figure~\ref{fig:increase_sv}. From the figure, we can see that structural virality increases monotonically. However, this result is not globally optimal, since a path graph should have a larger structural virality value. This suggests that the proposed first-order approximation method may have difficulty reaching the global optimum because it ignores higher-order terms in Eq.~\ref{eq:perturb_L}.

From another perspective, as structural virality increases, the tree tends to evolve toward a more real-like propagation pattern. In particular, the max-breadth decreases and the depth gap becomes smaller. These factors may be important structural cues for fake news detection.


\end{document}